%% file: 0_main_lcss.tex
\PassOptionsToPackage{dvipsnames,svgnames,table}{xcolor}

\documentclass[journal,twoside,web]{ieeecolor}
\usepackage{lcsys}

\usepackage[backend=biber,style=ieee,maxbibnames=6]{biblatex}
\usepackage{comment}
\usepackage{amsmath,amssymb,amsfonts}
\usepackage{graphicx}
\usepackage{textcomp}
\usepackage{xcolor}
\usepackage{array}
\usepackage{tabularx}
\usepackage{booktabs}
\usepackage{multirow}
\usepackage{subcaption}
\usepackage{capt-of}

\usepackage{placeins}
\usepackage{stfloats}

\usepackage{algorithm}
\usepackage{algpseudocode}

\usepackage{tikz}
\usetikzlibrary{arrows.meta,positioning,calc}

\usepackage[hidelinks]{hyperref}

\usepackage{amsthm}

\makeatletter
\@ifundefined{IEEEproof}{}{}
\makeatother

\theoremstyle{plain}
\newtheorem{theorem}{Theorem}
\newtheorem{proposition}{Proposition}

\newtheorem{corollary}{Corollary}
\newtheorem{assumption}{Assumption}

\theoremstyle{remark}
\newtheorem{remark}{Remark}
\theoremstyle{definition}
\newtheorem{definition}{Definition}

\let\labelindent\relax
\usepackage{enumitem}

\newcommand{\zc}{z^{c}}
\newcommand{\Gbar}{\bar{\Gamma}}
\newcommand{\Exp}{\mathbb{E}}
\newcommand{\Prob}{\mathbb{P}}
\newcommand{\kexc}{\kappa_{\mathrm{exc}}}

\newcommand{\sgn}{\operatorname{sgn}}
\newcommand{\Gmag}{\mathcal{G}_{\mathrm{mag}}}
\newcommand{\Gsgn}{\mathcal{G}_{\mathrm{sgn}}}
\newcommand{\Var}{\operatorname{Var}}
\newcommand{\Cov}{\operatorname{Cov}}
\newcommand{\tr}{\operatorname{tr}}
\newcommand{\xs}{\hat{x}^{s}}
\newcommand{\xa}{\hat{x}^{a}}
\newcommand{\Afil}{A_{\mathrm{f}}}
\newcommand{\Qtail}{\mathcal{Q}}
\newcommand{\vsg}{\sigma}
\newcommand{\Lyap}{\mathcal{L}}

\long\def\rev#1{#1}

\newcolumntype{P}[1]{>{\centering\arraybackslash}p{#1}}
\def\BibTeX{{\rm B\kern-.05em{\sc i\kern-.025em b}\kern-.08em
    T\kern-.1667em\lower.7ex\hbox{E}\kern-.125emX}}

\begin{document}

\title{Distribution-Free Budgeted Stealthy Attack Scheduling for Remote State Estimation}

\author{Qazi Mairaj ud din, Sidra Ghayour Bhatti, and Qadeer Ahmed
\thanks{The authors are with the Center for Automotive Research, The Ohio
State University, Columbus, OH, USA (e-mail: mairajuddin.1@osu.edu;
bhatti.39@osu.edu; ahmed.358@osu.edu).}
\thanks{\textcopyright~2026 IEEE. Personal use of this material is permitted. Permission from IEEE must be obtained for all other uses, in any current or future media, including reprinting/republishing this material for advertising or promotional purposes, creating new collective works, for resale or redistribution to servers or lists, or reuse of any copyrighted component of this work in other works. DOI: 10.1109/LCSYS.2026.3749337}}

\maketitle
\thispagestyle{empty}

\begin{abstract}
This letter addresses budgeted stealthy false-data-injection (FDI)
scheduling against remote state estimation, where a resource-constrained adversary may corrupt at most a fraction $\Gbar$ of transmissions. Existing event-triggered schedulers invert a Gaussian innovation tail to set the firing threshold and certify stealth by covariance matching; both are exact only under Gaussianity, which real cyber-physical residuals routinely violate. We propose a distribution-free scheduler
pairing the worst-case FDI action with a split-conformal calibrated
trigger, requiring neither the plant matrices nor any distributional model. We establish exact pathwise stealth against every magnitude-measurable detector, for any firing rule and innovation law; a finite-sample distribution-free bound on the mean firing rate, with almost-sure budget attainment under stationarity and ergodicity; and a steady-state degradation identity linear in a single scalar energy capture $\Psi$, maximized by the same order statistic that delivers the
budget guarantee. A conditional sign-symmetry condition delimits when the certificate extends to sign-sensitive detectors, the residual exposure being governed by a fourth cumulant that vanishes under Gaussian noise. Monte-Carlo studies and a
heavy-duty-truck CAN record confirm the bounds and quantify what the
Gaussian assumption costs outside its regime.
\end{abstract}


\input{1_lcss_rev4_b}
\input{1_experiment_v4b}

\printbibliography

\end{document}

%% file: 1_lcss_rev4_b.tex
\section{Introduction}
\IEEEPARstart{S}{tealthy} false data injection (FDI) is the most
consequential class of deception attack on cyber-physical systems
\cite{zhou2024cybersecurity}: it degrades estimation performance while evading the residual-based detectors guarding the sensor-to-estimator link. A practical adversary is also resource constrained, its bandwidth and energy limits expressed as a budget $\Gbar\in(0,1)$ on the fraction of time steps it may corrupt. This motivates \emph{event-based} scheduling, in which the adversary acts only after an informative event. The principle originates in benign sensor scheduling \cite{han2015stochastic}, was carried to energy-constrained DoS scheduling \cite{zhang2015optimal}, and then to stealthy FDI by Guo \emph{et al.} \cite{Guo2023}, who fire on the real-time residual and maximize the trace of the remote error covariance, and Wang \emph{et al.} \cite{Wang2025}, who extend it to multi-sensor systems.

Two parametric ingredients recur throughout. The trigger threshold is
set by inverting the Gaussian tail of the whitened innovation so
that the budget is met, and stealth is certified by a
\emph{covariance-matching} constraint requiring the corrupted innovation
to retain its nominal covariance $S$. Neither guarantee survives
departures from Gaussianity: off-Gaussian, the tail inversion misses the
budget, and a matched covariance no longer implies a matched law. Real
CPS residuals are routinely non-Gaussian --- nonlinearities inject
higher moments, saturation truncates a tail, mode switches produce
transients, and outliers, quantization and packet drops generate
occasional large innovations \cite{zhao2021event}. Vehicle CAN streams
are a concrete instance. 


This letter studies budgeted stealthy scheduling with the distributional
assumption removed. The adversary chooses the signal it
injects (Sec.~\ref{sec:signal}) and the instants at which to inject it (Sec.~\ref{sec:schedule}), giving the following contributions.


\begin{enumerate}
\item 
\rev{We certify the sign-flip FDI attack as exactly and pathwise stealthy against the class $\Gmag$ of magnitude-measurable detectors, under no distributional, independence or stationarity assumption and with no knowledge of the plant, filter or detector (Thm.~\ref{thm:magnitude}). The covariance certificate of \cite{Guo2023,Wang2025} is then shown necessary but not sufficient off Gaussianity: demanding the received law rather than its second moment collapses their feasible family to the sign-flip, which is
independently that family's damage maximizer  (Prop.~\ref{prop:secondmoment}).}
\item 
\rev{We propose a trigger replacing the Gaussian tail inverse with an order statistic of the nominal magnitude record, and prove a \emph{finite-sample} rate bound under exchangeability together with almost-sure budget attainment under stationarity and ergodicity alone Thm.~\ref{thm:budget}.}
\item 
\rev{We prove the adversary's entire influence collapses into one estimable scalar $\Psi$, the fraction of innovation energy the schedule
captures, separating firing rate from selection quality (Thm.~\ref{thm:degradation}), and that $\Psi$ is maximized causally by a
memoryless threshold on $|z_{k}|$ whose cut point is \emph{the same
order statistic} that delivers the budget guarantee, so one statistic
carries both (Cor.~\ref{cor:optimal}); \cite{Guo2023,Wang2025} are recovered as the Gaussian case.}
\item 
\rev{We identify exactly when the certificate extends to the sign-sensitive class $\Gsgn$: when the innovation is conditionally sign symmetric, which can fail only off Gaussianity (Prop.~\ref{prop:boundary}) --- and show that failure is coupled to the damage, both living in the tail the optimal trigger fires on. Experiments on Gaussian and heavy-tailed regimes and a heavy-duty-truck CAN record confirm each bound and show the Gaussian assumption mispredicting damage by up to a factor of two at small budgets.}
\end{enumerate}

\emph{Notation:} $\kappa_r(X)$ is the $r$-th cumulant of $X$ and
$\kappa_{\rm exc}(X):=\kappa_4(X)/\operatorname{Var}(X)^2$ its excess kurtosis, zero for a
Gaussian; $\operatorname{cum}(\cdot)$ is the joint cumulant; $\mathbf 1\{\cdot\}$ the
indicator; $\hat F_N$ the empirical CDF from $N$ samples; and, for $\rho(A)<1$,
$\mathcal L(A,W)$ the unique solution of $X=AXA^\top+W$.

\section{Problem Setup}
\label{sec:setup}

\subsection{Plant and Remote Estimation Architecture}
Consider the discrete-time linear plant
\begin{equation}
x_{k+1}=Ax_{k}+w_{k},\qquad y_{k}=Cx_{k}+v_{k},
\label{eq:plant}
\end{equation}
with $x_{k}\in\mathbb{R}^{n}$, $y_{k}\in\mathbb{R}^{m}$, and mutually
independent zero-mean i.i.d.\ noises $w_{k},v_{k}$ of covariance
$Q\succ0$, $R\succ0$. \emph{Neither is assumed Gaussian.} With $(A,C)$
detectable and $(A,Q^{1/2})$ stabilisable, the algebraic Riccati
equation has a unique stabilising solution $\bar P$; set
$S:=C\bar PC^{\top}+R$, $K:=\bar PC^{\top}S^{-1}$ and $\Afil:=A(I-KC)$,
which is Schur. A smart sensor collocated with \eqref{eq:plant} runs the
steady-state Kalman filter, producing $\xs_{k}$ and the innovation
\begin{equation}
z_{k}=y_{k}-C\xs_{k|k-1},\qquad \Exp[z_{k}z_{k}^{\top}]=S,
\label{eq:innov}
\end{equation}
whose prior error $e_{k}:=x_{k}-\xs_{k|k-1}$ obeys
\begin{equation}
e_{k+1}=\Afil e_{k}+w_{k}-AKv_{k},\qquad z_{k}=Ce_{k}+v_{k}.
\label{eq:errdyn}
\end{equation}
The sensor transmits $z_{k}$ over a network and the remote estimator
applies the received innovation $\zc_{k}$ directly,
\begin{equation}
\xa_{k}=\xa_{k|k-1}+K\zc_{k},\qquad \xa_{k+1|k}=A\xa_{k}.
\label{eq:remote}
\end{equation}

\begin{assumption}
\label{as:A1}
\emph{(i)} The network carries $z_{k}$ and the remote node updates by
\eqref{eq:remote} \cite[Rem.~1]{Guo2023}. \emph{(ii)} The sensor
computes \eqref{eq:innov} from the true measurements, and $\xa_{k}$ is
fed back neither to the plant nor to the sensor filter
\cite[Assum.~2]{Wang2025}. \rev{\emph{(iii)} $m=1$, so that the score, the level set of Cor.~\ref{cor:optimal} and the feasible family of Prop.~\ref{prop:boundary} are scalar; Rem.~\ref{rem:m1} discusses $m>1$.} 
\end{assumption}

Under \ref{as:A1}(ii) the monitored signal is \emph{exogenous} to the
attack: corruption accumulates at the remote node while the signal being
corrupted is generated from the nominal stream, so $\Afil$ is the
\emph{local filter} error matrix and no loop is closed around the plant.

\subsection{Signs and Magnitudes Off Gaussianity}
\label{sec:signs}

By \eqref{eq:errdyn} the innovation is white,
$\Exp[z_{k}z_{k+\ell}^{\top}]=0$ for $\ell\ge1$, whatever the noise law.
\rev{Under Gaussian noise whiteness upgrades to independence and signs become
independent of magnitudes. Off Gaussianity the innovation is
uncorrelated but not independent, and the two are linked. The attack
below acts on signs alone, so the strength of that link governs what it
can and cannot evade.}

\begin{definition}
\label{def:css}
\rev{$\{z_{k}\}$ is \emph{conditionally sign symmetric} if, given the entire
magnitude process $\{|z_{j}|\}_{j}$, the signs $\{\sgn(z_{k})\}$ are
i.i.d.\ uniform on $\{\pm1\}$; in particular
$\Exp[\sgn(z_{k})\mid\mathcal{F}_{k-1},|z_{k}|]=0$, where
$\mathcal{F}_{k-1}=\sigma(z_{j}:j\le k-1)$.}
\end{definition}

Definition~\ref{def:css} holds automatically under Gaussian noise. 

\subsection{Detection and Threat Models}
\label{sec:threat}
The adversary does not know which detector runs at the remote node, so
we fix a class rather than a detector. A residual-based detector is a
measurable functional $g$ of $\{\zc_{j}\}_{j\le k}$ raising an alarm
when $g>\eta$, with $\eta$ calibrated on nominal data to a false-alarm
rate $\epsilon$. \rev{Let $\Gmag$ be those detectors measurable with respect
to the magnitude process $\sigma(\|\zc_{j}\|:j\le k)$ --- the memoryless
$\chi^{2}$ test $\zc_{k}{}^{\top}S^{-1}\zc_{k}>\eta$, its windowed and
CUSUM variants, and the Serial Detector of \cite{bonczek2021detection}
all lie in $\Gmag$ --- and $\Gsgn$ those measurable with respect to the full received sequence
but not its magnitudes alone: the innovation-whiteness tests standard in
fault detection --- lag-$\ell$ autocorrelation and Ljung--Box --- and sign-balance tests on $\sgn(\zc_{k})$.}

The adversary is a man-in-the-middle on the sensor-to-estimator link. It
observes the nominal transmitted stream for $N$ steps and may replace
the transmitted packet during the attack window. It has no access to
$(A,C,Q,R)$, $S$ or $K$, to the remote estimator's state, or to $g$ and
$\eta$, and does not know the innovation law. With
$\gamma_{k}\in\{0,1\}$ the firing indicator, the budget is
\begin{equation}
\limsup_{T\to\infty}\tfrac{1}{T}\textstyle\sum_{k=1}^{T}\gamma_{k}\;\le\;\Gbar .
\label{eq:budget}
\end{equation}
The problem is to design a causal trigger satisfying \eqref{eq:budget}
that maximizes degradation of $\xa_{k}$ and evades the detector class,
without the plant matrices and without any distributional assumption.

\section{The Attack Signal}
\label{sec:signal}

\subsection{The Worst-Case FDI Action}
\label{sec:action}
At firing instants the adversary transmits
\begin{equation}
\zc_{k}=(1-2\gamma_{k})\,z_{k}=
\begin{cases}
-z_{k}, & \gamma_{k}=1,\\
\phantom{-}z_{k}, & \gamma_{k}=0 .
\end{cases}
\label{eq:action}
\end{equation}
\rev{Two features distinguish this \emph{flip} from the measurement-space
injections common in the FDI literature. Under Assum.~\ref{as:A1}(i) it is \emph{model-free by construction}: the adversary negates a signal it already reads, so no prediction, filter copy or covariance is required. And it coincides with the model-aware optimum: under the linear attack family
\begin{equation}
z^{a}_{k}=F_{k}z_{k}+b_{k},\qquad
b_{k}\sim\mathcal{N}(0,\Sigma_{k})\ \text{independent of}\ z_{k},
\label{eq:family}
\end{equation}
subject to the \emph{covariance-matching} stealth constraint
\begin{equation}
F_{k}SF_{k}^{\top}+\Sigma_{k}=S,
\label{eq:guoconstraint}
\end{equation}
the maximizer of the remote error covariance is $F^{\star}_{k}=-I$,
$\Sigma^{\star}_{k}=0$ \cite[Thm.~3]{Guo2023}. This section shows the
\emph{certificate} \eqref{eq:guoconstraint} does not survive off
Gaussianity, while the action does.}

\subsection{Exact Stealth Against $\Gmag$}
\label{sec:magnitude}

\begin{theorem}
\label{thm:magnitude}
\rev{Under \eqref{eq:action}, $\|\zc_{k}\|=\|z_{k}\|$ for every $k$ on every
sample path, irrespective of the firing rule, the budget $\Gbar$, and
the distribution of $z_{k}$. Consequently every $g\in\Gmag$ satisfies
$g(\{\zc_{j}\})=g(\{z_{j}\})$ pathwise, and for every $\eta$ its
false-alarm rate under attack equals its nominal rate exactly.}
\end{theorem}

\begin{proof}
$\zc_{k}=\vsg_{k}z_{k}$ with $\vsg_{k}\in\{-1,+1\}$, so
$\|\zc_{k}\|=\|z_{k}\|$; any $g\in\Gmag$ is a measurable function of
$\{\|\zc_{j}\|\}_{j\le k}=\{\|z_{j}\|\}_{j\le k}$ and so takes identical
values on the two paths. Pathwise equality of $\mathbf1\{g>\eta\}$
implies equality of its expectation.
\end{proof}

\begin{remark}[Scope and reach]
\label{rem:scope}
\rev{No knowledge of $S$, $\eta$, $g$ or the filter enters
Thm.~\ref{thm:magnitude}. It is also what makes exact stealth compatible
with $\Gbar<1$: the identity holds on firing \emph{and} non-firing
steps, so no signature accumulates on the instants skipped. The
certificate does not reach a detector reading the received \emph{signs};
Sec.~\ref{sec:boundary} determines when it extends to $\Gsgn$.}
\end{remark}

\rev{Since $\|\zc_{k}\|=\|z_{k}\|$, a score built from magnitudes is
computable from the adversary's \emph{output} stream and the attack
cannot corrupt its own trigger; and since the certificate is a magnitude
identity, a rule that fires on signs leaves the guaranteed class. Every
score below is therefore a function of $|z_{k}|$.}


\subsection{Covariance Matching Is Necessary, Not Sufficient}
\label{sec:secondmoment}

\rev{A detector calibrated on nominal data is defeated when the received
sequence is indistinguishable from the nominal one. Within a jointly
Gaussian family the law is fixed by the second moment, so
\eqref{eq:guoconstraint} \emph{is} that requirement and needs no further
justification in \cite{Guo2023,Wang2025}. Off Gaussianity it is strictly
weaker.} Writing $F_{k}=f$, $\Sigma_{k}=\sigma_{b}^{2}$, constraint
\eqref{eq:guoconstraint} reads $f^{2}S+\sigma_{b}^{2}=S$, so the
feasible set is the one-parameter family
$\{|f|\le1,\ \sigma_{b}^{2}=S(1-f^{2})\}$.

\begin{proposition}
\label{prop:secondmoment}
\rev{Let $(f,\sigma_{b}^{2})$ be feasible for \eqref{eq:guoconstraint} with
$b_{k}$ Gaussian. Then:}
\begin{enumerate}
\rev{\item[(i)] every feasible pair satisfies $\Var(z^{a}_{k})=S$, while
  \begin{equation}
  \kexc(z^{a}_{k})=f^{4}\,\kexc(z_{k});
  \label{eq:kurtlaw}
  \end{equation}
\item[(ii)] if $\kexc(z_{k})\neq0$, every feasible pair with $|f|<1$ is
  distinguishable from nominal by a fourth-moment test; if in addition
  $z_{k}$ is symmetric, the only law-preserving feasible pairs are
  $f=+1$ (no attack) and $f=-1$, $\sigma_{b}^{2}=0$ (the flip);
\item[(iii)] independently of \emph{(i)}--\emph{(ii)}, the error
  $z_{k}-z^{a}_{k}$ injected into \eqref{eq:remote} at a firing instant
  has second moment $2S(1-f)$, strictly decreasing in $f$ 
  and uniquely maximized at $f=-1$, where it equals $4S$.}
\end{enumerate}
\end{proposition}

\begin{proof}
(i) By additivity and degree-$r$ homogeneity of cumulants,
$\kappa_{4}(z^{a})=f^{4}\kappa_{4}(z)$, the Gaussian term vanishing; the
constraint gives $\Var(z^{a})=S$, whence \eqref{eq:kurtlaw}.
(ii) 
For $|f|<1$ we have $f^{4}<1$, so by \eqref{eq:kurtlaw} the laws
differ in their fourth cumulant, by $(1-f^{4})|\kexc(z_{k})|$; law
preservation forces equality of fourth cumulants, hence $f^{4}=1$,
$|f|=1$, $\sigma_{b}^{2}=0$ and $z^{a}=\pm z$, and symmetry gives
$-z\overset{d}{=}z$.
(iii) $(1-f)z_{k}-b_{k}$ has second moment
$(1-f)^{2}S+S(1-f^{2})=2S(1-f)$.
\end{proof}

\rev{Every member of the family passes the covariance check, yet each distorts the fourth cumulant by a known factor, so a defender who looks past the second moment sees all of them, except the flip.}

\begin{remark}
\label{rem:hinge}
\rev{Under Gaussianity $\kexc(z)=0$, \eqref{eq:kurtlaw} is vacuous and
\eqref{eq:guoconstraint} suffices. Nor can a non-Gaussian $b_{k}$ remove
the obstruction: matching the received law at some $|f|<1$ would require
an independent $b$ with $\kappa_{r}(b)=(1-f^{r})\kappa_{r}(z)$ for all
$r\ge3$, i.e.\ \emph{self-decomposability} of the innovation law
\cite{ken1999levy}, which the adversary can neither verify nor exploit
without modeling that law. The Gaussian law being self-decomposable, the
classical treatment never meets the obstruction.}
\end{remark}

\rev{The signal is \emph{determined}, and every remaining degree
of freedom lies in the schedule $\{\gamma_{k}\}$.}

\section{The Schedule}
\label{sec:schedule}

\subsection{Distribution-Free Budget Calibration}
\label{sec:calibration}

Let $s_{k}$ be a causal score that is a function of magnitudes alone
with continuous marginal CDF $F_{s}$. The canonical choice, adopted
throughout, is $s_{k}=|z_{k}|$, which the adversary reads directly from
its own stream. Given a nominal calibration record $s_{1},\dots,s_{N}$
with order statistics $s_{(1)}\le\dots\le s_{(N)}$ and a budget
$\Gbar\ge 1/(N+1)$, set
\begin{equation}
\hat\varepsilon_{N}=s_{(j_{N})},\quad
j_{N}:=\left\lceil (N+1)(1-\Gbar)\right\rceil,\quad
\gamma_{k}=\mathbf1\{s_{k}>\hat\varepsilon_{N}\}.
\label{eq:conformal}
\end{equation}
The condition $\Gbar\ge1/(N+1)$ is what makes $j_{N}\le N$. If it fails,
the budget is unresolvable from $N$ samples and we set
$\hat\varepsilon_{N}=+\infty$, which never fires.

\begin{theorem}
\label{thm:budget}
\rev{Consider the trigger \eqref{eq:conformal}.}
\begin{enumerate}
\rev{\item[(i)] If $s_{1},\dots,s_{N+1}$ are exchangeable, then
  \begin{equation}
  \Prob\!\left(s_{N+1}>\hat\varepsilon_{N}\right)
  =1-\frac{j_{N}}{N+1}\;\le\;\Gbar ,
  \label{eq:conformalbound}
  \end{equation}
  and if every deployment score is exchangeable with the calibration
  record, $\Exp\big[T^{-1}\sum_{k\le T}\gamma_{k}\big]\le\Gbar$ at every
  horizon $T$.}
\item[(ii)] If $\{s_{k}\}$ is stationary and ergodic, then for fixed
  $N$, $T^{-1}\sum_{k\le T}\gamma_{k}\to1-F_{s}(\hat\varepsilon_{N})$
  a.s.; and if $F_{s}$ is strictly increasing at $F_{s}^{-1}(1-\Gbar)$,
  then $\hat\varepsilon_{N}\to F_{s}^{-1}(1-\Gbar)$ and
  $1-F_{s}(\hat\varepsilon_{N})\to\Gbar$ a.s.\ as $N\to\infty$.
\end{enumerate}
\end{theorem}
\begin{proof}
(i) Continuity of $F_{s}$ makes the $N+1$ scores a.s.\ distinct, so
under exchangeability the rank of $s_{N+1}$ is uniform and
$\Prob(s_{N+1}>s_{(j)})=1-j/(N+1)$; substituting
$j=j_{N}\ge(N+1)(1-\Gbar)$ gives \eqref{eq:conformalbound}, and
linearity of expectation the second claim.
(ii) Birkhoff's theorem applied to
$\mathbf1\{s_{k}>\hat\varepsilon_{N}\}$ gives the first limit; the
ergodic Glivenko--Cantelli theorem \cite{tucker1959generalization},
continuity of $F_{s}^{-1}$ at $1-\Gbar$, the continuous mapping theorem
\cite{vd1998asymptotic}, and $j_{N}/(N+1)\to1-\Gbar$ the second.
\end{proof}

\begin{remark}
\label{rem:index}
\rev{Part (ii) delivers \eqref{eq:budget} in the almost-sure sense in which
it is stated; part (i) bounds the expected rate at every finite horizon.
The one-index change from the naive
$\inf\{t:\hat F_{N}(t)\ge1-\Gbar\}$, which calibrates on the $N$
calibration points rather than the $N+1$ including the fresh one and so
over-fires, is split-conformal calibration~\cite{lei2018distribution}. Two qualifications delimit
it. First, \eqref{eq:conformalbound} bounds the \emph{mean} realized
rate: the conformal index is near median-unbiased, so the budget is
still exceeded on a substantial fraction of calibration draws, and a
training-conditional guarantee would need a more conservative index.
Second, part (i) assumes exchangeability, which the magnitude sequence
satisfies only approximately, being stationary but serially dependent;
part (ii) needs only stationarity and ergodicity and pays for dependence
in variance rather than bias. The \emph{Budget Calibration} results of
Sec.~\ref{sec:results}(c) validate both.}
\end{remark}

\begin{remark}
\label{rem:m1}
\rev{For $m=1$ the model-aware statistic
$\|S^{-1/2}z_{k}\|_{\infty}=|z_{k}|/\sqrt{S}$ and the model-free
magnitude $|z_{k}|$ differ by a positive factor, which preserves the
ordering of the calibration sample. Since an empirical quantile depends
on the sample only through that ordering, the two rules induce
\emph{identical} firing sets whenever both thresholds are set by
\eqref{eq:conformal}. Granting the adversary the plant model therefore
contributes nothing to budget accuracy in the scalar case: the binding
assumption in \cite{Guo2023,Wang2025} is the Gaussian tail inversion,
not the model. For $m>1$ the matrix $S^{-1/2}$ mixes channels, is no
longer monotone in $|z_{k}|$, and genuinely reorders firing instants.}
\end{remark}

\subsection{Distribution-Free Degradation}
\label{sec:degradation}

Let $d_{k}:=\xs_{k}-\xa_{k}$ be the divergence between the clean local
estimate and the corrupted remote one. Subtracting \eqref{eq:remote}
from the sensor recursion and using $z_{k}-\zc_{k}=2\gamma_{k}z_{k}$,
\begin{equation}
d_{k}=Ad_{k-1}+K(z_{k}-\zc_{k})=Ad_{k-1}+2\gamma_{k}Kz_{k}:
\label{eq:divergence}
\end{equation}
each firing injects a kick $2Kz_{k}$ proportional to the innovation at
that instant, and past kicks decay through $A$. With
$D_{k}:=\Exp[d_{k}d_{k}^{\top}]$, expanding \eqref{eq:divergence} gives
\begin{multline}
D_{k}=AD_{k-1}A^{\top}+4K\,\Exp[\gamma_{k}z_{k}z_{k}^{\top}]K^{\top}\\
+2A\,\Exp[d_{k-1}\gamma_{k}z_{k}^{\top}]K^{\top}+(\cdot)^{\top},
\label{eq:Drec}
\end{multline}
whose cross terms vanish under Def.~\ref{def:css}: $d_{k-1}$ is
$\mathcal{F}_{k-1}$-measurable and $\gamma_{k}$ is a function of
magnitudes, so conditioning on $\mathcal{F}_{k-1}$ and $|z_{k}|$ leaves
$\Exp[\sgn(z_{k})\mid\cdot]=0$; Sec.~\ref{sec:results}(f) measures what
survives when Def.~\ref{def:css} fails.

We take $\tr(D_{k})$ as the measure of attack effect: it is the
component of the remote error attributable to the attack and vanishes in
its absence. It also avoids a hypothesis $\tr(P^{a}_{k})$ would require,
namely independence of $e^{s}_{k}$ from the innovation sequence, which
off Gaussianity does not follow from the Kalman orthogonality
$\Exp[e^{s}_{k}z_{j}^{\top}]=0$ because $\gamma_{j}$ acts nonlinearly on
$z_{j}$. 

\begin{theorem}
\label{thm:degradation}
\rev{Let $\rho(A)<1$, let $\{z_{k}\}$ be stationary and satisfy
Def.~\ref{def:css}, let $\gamma_{k}$ be a function of magnitudes alone
with $\Prob(\gamma_{k}=1)=\Gbar$, and let $D_{0}=0$. With the \emph{energy capture} $\Psi:=\Exp[\gamma_{k}z_{k}^{2}]/S$,
\begin{multline}
D_{\infty}=AD_{\infty}A^{\top}+4\Psi S KK^{\top}, \\ \tr(D_{\infty})=4S\Psi\,\tr\Lyap(A,KK^{\top}),
\label{eq:Dinf}
\end{multline}

and $\hat\Psi_{N}:=\sum_{k\le N}\gamma_{k}s_{k}^{2}/\sum_{k\le N}s_{k}^{2}$
is strongly consistent for $\Psi$ under the hypotheses of
Thm.~\ref{thm:budget}(ii).}
\end{theorem}

\begin{proof}
Stationarity makes the driving term of \eqref{eq:Drec} equal to
$4\Psi S KK^{\top}$ at every $k$. With the cross terms removed and
$\rho(A)<1$ the recursion converges to the unique solution of the
discrete Lyapunov equation, whose linearity in the driving term gives
the trace identity. Consistency follows from Birkhoff's theorem applied
to $\gamma_{k}z_{k}^{2}$ and $z_{k}^{2}$, with
Thm.~\ref{thm:budget}(ii) supplying convergence of the threshold.
\end{proof}


Since $4S$ and $\tr\Lyap(A,KK^{\top})$ are properties of the plant and
filter that no schedule can alter, the adversary's entire influence is
isolated in the single scalar $\Psi$ --- the share of innovation energy
it gets to corrupt.

\subsection{The Optimal Schedule}
\label{sec:optimal}

By \eqref{eq:Dinf} the design problem
\begin{equation}
\max_{\gamma}\ \tr(D_{\infty})
\quad\text{s.t.}\quad \Prob(\gamma_{k}=1)=\Gbar
\label{eq:designproblem}
\end{equation}
is equivalent to maximizing $\Psi$ alone. Here $\gamma$ is confined to $\Gmag$, which by Sec.~\ref{sec:magnitude} keeps the schedule inside the class the action is certified against.

\begin{corollary}
\label{cor:optimal}
\rev{Let $F_{|z|}$ be continuous. Then:}
\begin{enumerate}
\rev{\item[(i)] for every magnitude-measurable rule of rate $\Gbar$,
  $\Psi=\Gbar+\Cov(\gamma_{k},z_{k}^{2})/S$, so the budget fixes the
  first term and a schedule competes only through the second;
\item[(ii)] $\Psi$ is maximized by the memoryless upper level set
  \begin{equation}
  \gamma^{\star}_{k}=\mathbf1\{|z_{k}|>q\},
  \qquad q:=F^{-1}_{|z|}(1-\Gbar),
  \label{eq:optrule}
  \end{equation}
  and any maximizer agrees with $\gamma^{\star}$ almost everywhere;
  under the hypotheses of Thm.~\ref{thm:degradation}, $\gamma^{\star}$
  therefore solves \eqref{eq:designproblem};
\item[(iii)] at the optimum
  $\Psi^{\star}=\Exp[z_{k}^{2}\mathbf1\{|z_{k}|>q\}]/S$, the share of
  the second moment carried by the upper $\Gbar$-tail of $|z_{k}|$.}
\end{enumerate}
\end{corollary}

\begin{proof}
(i) $\Exp[\gamma z^{2}]=\Exp[\gamma]\Exp[z^{2}]+\Cov(\gamma,z^{2})$.
(ii) With $p(|z_{k}|):=\Exp[\gamma_{k}\mid|z_{k}|]\in[0,1]$,
$\Exp[p]=\Gbar$, and $p^{\star}=\mathbf1\{|z|>q\}$, the sign of
$p-p^{\star}$ opposes that of $z^{2}-q^{2}$ pointwise, so
$\Exp[(p-p^{\star})z^{2}]\le q^{2}\Exp[p-p^{\star}]=0$, with equality
only if $p=p^{\star}$ off $\{|z|=q\}$; optimality for
\eqref{eq:designproblem} follows from \eqref{eq:Dinf}, increasing in
$\Psi$. (iii) Substitute $\gamma^{\star}$ into $\Psi$.
\end{proof}

\rev{Equation \eqref{eq:optrule} is the trigger \eqref{eq:conformal} with the
population quantile in place of the order statistic, and
$\hat\varepsilon_{N}$ is a consistent estimate of $q$ by
Thm.~\ref{thm:budget}(ii): \emph{one order statistic delivers both the
finite-sample budget guarantee and the damage optimum}, at one
comparison per step and with neither the plant matrices nor any
distributional assumption (Alg.~\ref{alg:main}). The Gaussian
tail-inversion rule of \cite{Guo2023,Wang2025} is also an upper level
set of $|z_{k}|$, so the two differ only in realized rate --- and by (i),
which separates rate from selection quality, schedulers must therefore
be compared at $\Psi$ rather than on damage.}

\begin{remark}
\label{rem:gaussian}
\rev{Let $z_{k}\sim\mathcal N(0,S)$ be i.i.d. Def.~\ref{def:css} then holds,
so by Prop.~\ref{prop:boundary} below the flip is stealthy against
\emph{every} residual-based detector. The threshold of
\eqref{eq:optrule} acquires the closed form
$q=\sqrt{S}\,\Qtail^{-1}(\Gbar/2)$, so $\hat\varepsilon_{N}$ converges
to the Gaussian tail-inversion threshold of \cite[Thm.~1]{Guo2023}; and
with $c:=\Qtail^{-1}(\Gbar/2)$ and $\varphi$ the standard normal
density, the energy capture takes the Gaussian truncated-second-moment
form $\Psi_{\mathcal G}=\Gbar+2c\,\varphi(c)=1-(1-\Gbar)(1-\beta(c))$
with $\beta(c)=2c\varphi(c)/(1-\Gbar)$ the truncation factor of
\cite[eq.~(44)]{Guo2023}. Thm.~\ref{thm:degradation} then reduces to the
closed-form degradation of \cite{Guo2023,Wang2025} and
Cor.~\ref{cor:optimal} to their optimal schedule, the attack-induced
excess being $4\,\Psi_{\mathcal G}\,\tr(KC\bar P)$ per step. This is the
single point at which the distributional assumption enters the damage
prediction, and the \emph{Damage} results of Sec.~\ref{sec:results}(f)
quantify what it costs outside its regime.}
\end{remark}

\begin{algorithm}[t]
\caption{Distribution-free budgeted attack scheduling}
\label{alg:main}
\begin{algorithmic}[1]
\Require budget $\Gbar$, calibration length $N\ge\Gbar^{-1}-1$
\State observe the nominal stream $z_{1},\dots,z_{N}$; $s_{k}\gets|z_{k}|$
\State $\hat\varepsilon_{N}\gets s_{(\lceil(N+1)(1-\Gbar)\rceil)}$
\State $\hat\Psi\gets
       \sum_{k\le N}\mathbf1\{s_{k}>\hat\varepsilon_{N}\}\,s_{k}^{2}
       \big/ \sum_{k\le N}s_{k}^{2}$
\For{$k=1,\dots,T$}
  \State transmit $\zc_{k}\gets-z_{k}$ if $|z_{k}|>\hat\varepsilon_{N}$,
         else $\zc_{k}\gets z_{k}$
\EndFor
\end{algorithmic}
\end{algorithm}

\subsection{The Boundary and the Damage--Exposure Coupling}
\label{sec:boundary}


\begin{proposition}
\label{prop:boundary}
\rev{Let the firing rule be a function of magnitudes alone,
$\gamma_k=\phi(\lvert z_{k-L}\rvert,\dots,\lvert z_k\rvert)$, and set
$A_k:=\{\gamma_k=1\}$. If $\{z_k\}$ is conditionally sign symmetric then
$\{z^c_k\}\overset{d}{=}\{z_k\}$ as processes, and every detector, in
$\mathcal G_{\rm mag}$ or $\mathcal G_{\rm sgn}$, retains its nominal false-alarm rate.
Conversely, Def.~\ref{def:css} fails whenever
\begin{equation}\label{eq:fourth}
\mathbb E[z_k^3z_{k+1}]
 = CA_f\operatorname{cum}(z_k,z_k,z_k,e_k) - CAK\,\kappa_4(v_k)\neq0,
\end{equation}
in which case the received sequence carries the induced lag-one autocovariance
\begin{equation}\label{eq:induced}
\mathbb E[z^c_kz^c_{k+1}] = -2\,(T_{10}+T_{01}),
\end{equation}
$T_{10}:=\mathbb E[\mathbf 1_{A_k}\mathbf 1_{A^c_{k+1}}z_kz_{k+1}]$ and
$T_{01}:=\mathbb E[\mathbf 1_{A^c_k}\mathbf 1_{A_{k+1}}z_kz_{k+1}]$, and detectors in
$\mathcal G_{\rm sgn}$ acquire power. Both terms in \eqref{eq:fourth} are fourth cumulants
and vanish for Gaussian $w,v$: under Gaussian noise the flip is stealthy against every
residual-based detector, and the exposure is a purely non-Gaussian phenomenon.}
\end{proposition}

\begin{proof}
Sufficiency: conditionally on the magnitude process the signs are
i.i.d.\ uniform and $\vsg_{k}$ is magnitude measurable, so
$\{\vsg_{k}\sgn(z_{k})\}$ has the same conditional law as
$\{\sgn(z_{k})\}$; since $\zc_{k}=\vsg_{k}\sgn(z_{k})|z_{k}|$ and the
magnitudes are unchanged, integrating gives
$\{\zc_{k}\}\overset{d}{=}\{z_{k}\}$. Necessity: Def.~\ref{def:css}
gives $\Exp[\sgn(z_{k})\sgn(z_{k+1})\mid\{|z_{j}|\}]=0$ and hence
$\Exp[z_{k}^{3}z_{k+1}]=0$; writing each expectation as a joint
cumulant plus its Gaussian part and using $e\perp v$, the Gaussian parts cancel by
$A_f\bar PC^\top=A(\bar PC^\top-K(S-R))=AKR$, which gives \eqref{eq:fourth}. Finally
$(1-2\gamma_k)(1-2\gamma_{k+1})=1-2\gamma_k-2\gamma_{k+1}+4\gamma_k\gamma_{k+1}$ and
$\mathbb E[z_kz_{k+1}]=0$ give \eqref{eq:induced}, the concordant terms canceling.
\end{proof}

\rev{The magnitude-measurability hypothesis is the constraint already imposed
in Sec.~\ref{sec:magnitude}. A sign-dependent rule such as
$\gamma_{k}=\mathbf1\{z_{k}>\varepsilon\}$ makes the received marginal
one sided in every regime, which under Gaussian innovations is invisible
to whiteness tests but is exposed by a sign-balance test, and off
Gaussianity destroys whiteness as well. Sec.~\ref{sec:results}(g) 
quantifies both.}


\begin{remark}
\label{rem:coupling}
\rev{Cor.~\ref{cor:optimal} and Prop.~\ref{prop:boundary} concern the same
quantity. If $\gamma_{k}$ is independent of $\{z_{j}\}$ then
$\Exp[\zc_{k}\zc_{k+\ell}]=0$ for every $\ell\ge1$ and every
autocorrelation detector in $\Gsgn$ keeps its nominal rate --- but then
$\Cov(\gamma_{k},z_{k}^{2})=0$ and $\Psi=\Gbar$, the value attained by
firing blind. At the other extreme $\gamma^{\star}$ concentrates the
firing set on the tail, which is exactly where the sign--magnitude
dependence of \eqref{eq:fourth} lives. Damage and residual
detectability are therefore two readings of one number, and off
Gaussianity the adversary cannot raise the first without raising the
second. Both terms of \eqref{eq:induced} require one instant to fire while its neighbor does not, so each vanishes as $\bar\Gamma\to0$ and as $\bar\Gamma\to1$: the induced autocovariance, and with it the power of a whiteness monitor, must attain an interior extremum in the budget. Scheduling under an explicit exposure constraint requires leaving the memoryless class of \eqref{eq:optrule} and is left to future work.}
\end{remark}

%% file: 1_experiment_v4b.tex
\section{Results}
\label{sec:results}

\paragraph{Setup}
\label{sec:setup-results}
Two regimes share the plant \eqref{eq:plant} with
$A=\left[\begin{smallmatrix}0.95&0.02\\0&0.90\end{smallmatrix}\right]$,
$C=[1\ \ 0]$, $Q=0.01I_{2}$ and $R=0.05$, so $m=1$, $\rho(A)=0.95$,
$S=0.0752$, $K=[0.335\ \ 0.029]^{\top}$, $\tr\Lyap(A,KK^{\top})=1.182$, and a nominal posterior error trace of $0.0690$. Under \textbf{N1} both
noises are Gaussian. Under \textbf{N2} the sensor noise is the mixture
$0.95\,\mathcal N(0,0.0227)+0.05\,\mathcal N(0,0.568)$, whose variance is
$0.0500=R$ by construction, so both regimes present the same $S$ and
distributional shape is the only manipulated variable. The filter gives
$\kexc(z)=7.49$ (heavy-tail) against $0.00$ under N1.

The third regime is real. \textbf{CSU} is a $15{,}123$\,s J1939 CAN log
from a heavy-duty truck~\cite{biggs2024modeling} on a $10$\,Hz grid. A
CAN bus carries measurements, not the transmitted innovation, so the
innovation is estimated without a plant model: $\Afil$ being Schur, the
innovation form admits an infinite-order VARX representation with
geometrically decaying coefficients~\cite{anderson2005optimal}, which we
fit by high-order least squares on nominal data. On ground-truth data
with the model hidden the fit recovers the innovation sequence with
correlation at least $0.993$ and $S$ to within $0.5\%$. Stationarity
masks leave $15$ segments and $118{,}081$ samples; the record gives
$\hat\kappa_{\rm exc}=28.5$.

Three schedulers are compared. \textbf{B1} inverts the Gaussian tail as
in \cite{Guo2023,Wang2025}. \rev{\textbf{B2} takes the empirical quantile of
the model-aware statistic $\|S^{-1/2}z_{k}\|_{\infty}$. \textbf{B3} is
Algorithm~\ref{alg:main}. Calibration uses $N=5000$ nominal samples,
deployment $T=50{,}000$ steps and $200$ Monte-Carlo realizations. On CSU each segment is split exchangeably and every scheduler is granted the exact
sample covariance $S$, so B1's error there is pure distributional shape.
Detector thresholds are calibrated empirically on a long nominal record
to a $1\%$ false-alarm rate.} 

\rev{\paragraph{Budget Accuracy}
\label{sec:res-budget}
Fig.~\ref{fig:budget} report the signed
error of the realized firing rate. Under N1 every scheduler tracks the
budget to within $0.06$~pp. B2 and B3 realize identical firing sets as $\max|\hat\Gamma_{\mathrm{B2}}-\hat\Gamma_{\mathrm{B3}}|=0$ over all $2400$ simulated comparisons, exactly as Rem.~\ref{rem:m1} predicts, so
B2 is not plotted separately. Under N2, B1 departs from the target by up to $8.76$~pp, a third of the budget in relative terms, and \emph{changes sign} between $\Gbar=0.02$ and $\Gbar=0.05$ --- the signature of leptokurtosis, the matched-variance mixture being simultaneously more peaked and heavier-tailed than the Gaussian. CSU sharpens the pattern ($+26\%$ at $\Gbar=0.02$, $-43\%$ at
$\Gbar=0.20$), and at $\Gbar=0.02$ B1 violates the hard budget
constraint \eqref{eq:budget} in $80\%$ of segments, precisely the
operating point an energy-limited adversary cares about. B1's error is bias: its Monte-Carlo spread is three times
\emph{tighter} than the empirical rules', since it has no estimation
step for more data to improve. Algorithm~\ref{alg:main} holds within
$\pm0.10$~pp on the simulated regimes and $1.14$~pp on CSU under no
distributional assumption, as Thm.~\ref{thm:budget}(ii) asserts.}

\rev{\paragraph{Budget Calibration}
\label{sec:calib}The one-index change of \eqref{eq:conformal} earns its place at short
records: at $N=100$, $\Gbar=0.02$ it realizes an exceedance rate of
$0.0198$ against $0.0293$ for the naive quantile. It violates the budget
on $41\%$ of draws against $67\%$, with the conformal fraction staying
in $[0.41,0.50]$ everywhere as Rem.~\ref{rem:index} predicts.  }

\rev{\paragraph{Stealth}\label{sec:stealth}
Theorem~\ref{thm:magnitude} is an algebraic identity, so an experiment
can only check that the implementation realizes it. Across seven
detectors --- memoryless and windowed $\chi^{2}$, CUSUM, CUSIGN, and the
magnitude, sign and combined components of the Serial Detector
\cite{bonczek2021detection} --- both regimes and all budgets, attacked and nominal false-alarm rates agree, $\max|\Delta\mathrm{FAR}|=0$ in double precision. The attack when applied through the
live pipeline, with $\gamma$ from the scheduler, $\zc$ to the detector, and $d_{k}$ driving the remote error, raising the remote mean-square error from $1.75\times$ nominal at $\Gbar=0.02$ to $5.75\times$ (N1) and $5.60\times$ (N2) at $\Gbar=0.50$.}

\rev{\paragraph{Covariance Gap}\label{sec:cov-gap} Proposition~\ref{prop:secondmoment} is visible in the same streams:
$\Var(z^{a})/S$ stays within $0.15\%$ of unity while $\kexc(z^{a})$
matches $f^{4}\kexc(z)$ to $1.2\%$, so \eqref{eq:guoconstraint} is exact
and blind to the fourth cumulant. A fourth-moment test then gains power
as $|f|\to0$, as Prop.~\ref{prop:secondmoment}(ii) implies: the maximum-entropy point $f=0$ is easiest to catch (power $0.78$ at $N=200$), so adding noise does not buy stealth and only the flip survives. }



\rev{\paragraph{Damage}
\label{sec:res-damage}
Fig.~\ref{fig:psi} compares $\Psi$ with the Gaussian closed form $\Psi_{\mathcal G}$ of Rem.~\ref{rem:gaussian} and with the plug-in $\hat\Psi_{N}$: the three agree under N1, which is the control, while under N2 $\Psi_{\mathcal G}$ understates $\Psi$ by a factor of $2.43$ at
$\Gbar=0.02$. 

Off Gaussianity the Def.~\ref{def:css} hypothesis of
Thm.~\ref{thm:degradation} no longer holds, so the cross term of
\eqref{eq:Drec} survives and \eqref{eq:Dinf} --- exact under N1 to within
$0.9\%$ --- over-predicts the realized excess mean-square error by
$5$--$31\%$ under N2, worst near $\Gbar=0.05$; the directly measurable
part of that gap is the surviving cross term, at $-29.0\%$ of the excess
at $\Gbar=0.02$ and decaying monotonically to $-0.6\%$ at $\Gbar=0.50$. 
This is the second reading of Rem.~\ref{rem:coupling}: the fourth cumulant that exposes the schedule to $\Gsgn$ is the same one that makes the Lyapunov identity inexact. The net effect on a Gaussian-calibrated adversary is nonetheless a large \emph{under}-estimate where it matters: at $\Gbar=0.02$  $\tr(D_{\infty}$) = $0.051$ against a realized $0.103$
.}

\begin{figure}[t]\centering
\subfloat[]{\includegraphics[width=0.49\columnwidth]{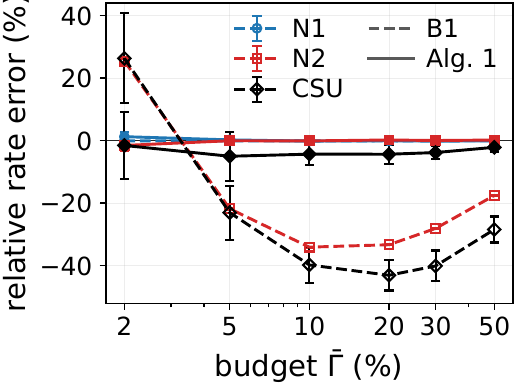}\label{fig:budget}}\hfil
\subfloat[]{\includegraphics[width=0.49\columnwidth]{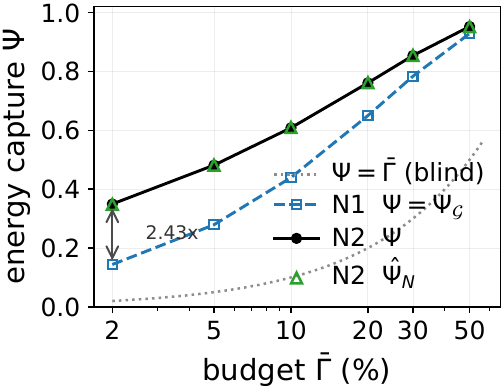}\label{fig:psi}}
\caption{(a) Signed relative error of the realized firing rate against the budget, $95\%$ confidence intervals; half-widths do not exceed
$0.10$ (N1), $0.11$ (N2) and $2.07$ (CSU) pp. (b) Energy capture $\Psi$ against the budget.}
\label{fig:main}
\end{figure}

\rev{\paragraph{Boundary}\label{sec:res-boundary}
A lag-one whiteness test on the optimal rule \eqref{eq:optrule}
separates the two regimes cleanly. Under N1, Def.~\ref{def:css} holds
and Prop.~\ref{prop:boundary} applies: the test never departs from its nominal level, with power at or below $0.018$ at every budget and
window, so the flip leaves no serial signature. Under N2 the hypothesis
fails, since \eqref{eq:optrule} fires on the largest $|z_{k}|$, which is
where the sign--magnitude dependence \eqref{eq:fourth} lives. The resulting exposure at $W=5000$ with the whiteness power peaks at $1.000$ for $\Gbar=0.02$, falling to $0.113$ at $\Gbar=0.50$ --- vanishing at both extremes exactly as Rem.~\ref{rem:coupling} predicts. The sign-dependent control
$\gamma_{k}=\mathbf1\{z_{k}>\varepsilon\}$ is invisible to the whiteness
test under N1, the innovation there being i.i.d., but a sign-balance test catches the sign-dependent control $\gamma_{k}=\mathbf1\{z_{k}>\varepsilon\}$ in both regimes.}

\rev{\section{Conclusion}
\label{sec:conclusion}

We have studied budgeted stealthy FDI scheduling against remote
estimation without the Gaussian innovation assumption underlying
existing event-triggered schedulers. Both parametric ingredients are
replaced by distribution-free objects: the tail inverse by a
split-conformal order statistic of the nominal magnitude record, and the
covariance-matching certificate by a pathwise magnitude identity. A
Gaussian model of a non-Gaussian residual therefore misleads in both
directions, missing the resource constraint by an irreducible bias and
understating the damage available to a low-budget adversary --- both
worst exactly where an energy-limited attacker operates.

Its price is that the sign-flip is no longer stealthy against
\emph{every} residual-based detector. Lifting that limitation --- to
detectors with memory and to vector measurements --- is left to future work. On the defender's side the boundary result is constructive: a monitor reading received signs, not magnitudes alone, breaks the certificate under non-Gaussian noise.}

%% file: bibliography_1.bib
@article{zhou2024cybersecurity,
  title={Cybersecurity landscape on remote state estimation: A comprehensive review},
  author={Zhou, Jing and Shang, Jun and Chen, Tongwen},
  journal={IEEE/CAA J. Autom. Sinica},
  volume={11},
  number={4},
  pages={851--865},
  year={2024},
  publisher={IEEE}
}

@article{han2015stochastic,
  author  = {Han, Duo and Mo, Yilin and Wu, Junfeng and Weerakkody, Sean and Sinopoli, Bruno and Shi, Ling},
  title   = {Stochastic Event-Triggered Sensor Schedule for Remote State Estimation},
  journal = {IEEE Trans. Autom. Control},
  volume  = {60},
  number  = {10},
  pages   = {2661--2675},
  year    = {2015},
  
}

@article{zhang2015optimal,
  title={Optimal denial-of-service attack scheduling with energy constraint},
  author={Zhang, Heng and Cheng, Peng and Shi, Ling and Chen, Jiming},
  journal={IEEE Trans. Autom. Control},
  volume={60},
  number={11},
  pages={3023--3028},
  year={2015},
  publisher={IEEE}
}

@book{vd1998asymptotic,
  title     = {Asymptotic Statistics},
  author    = {van der Vaart, A. W.},
  year      = {1998},
  publisher = {Cambridge University Press},
  series    = {}
}

@book{ken1999levy,
  title     = {L{\'e}vy Processes and Infinitely Divisible Distributions},
  author    = {Sato, Ken-Iti},
  volume    = {68},
  year      = {1999},
  publisher = {Cambridge University Press}
}

@article{lei2018distribution,
  title={Distribution-free predictive inference for regression},
  author={Lei, Jing and G’Sell, Max and Rinaldo, Alessandro and Tibshirani, Ryan J and Wasserman, Larry},
  journal={Journal of the American Statistical Association},
  volume={113},
  number={523},
  pages={1094--1111},
  year={2018},
  publisher={Taylor \& Francis}
}

@inproceedings{bonczek2021detection,
  title={Detection of hidden attacks on cyber-physical systems from serial magnitude and sign randomness inconsistencies},
  author={Bonczek, Paul J and Bezzo, Nicola},
  booktitle={2021 American Control Conference (ACC)},
  pages={3281--3287},
  year={2021},
  organization={IEEE}
}

@book{anderson2005optimal,
  title={Optimal filtering},
  author={Anderson, Brian DO and Moore, John B},
  year={2005},
  publisher={Courier Corporation}
}

@article{tucker1959generalization,
  title={A generalization of the Glivenko-Cantelli theorem},
  author={Tucker, Howard G},
  journal={The Annals of Mathematical Statistics},
  volume={30},
  number={3},
  pages={828--830},
  year={1959},
  publisher={JSTOR}
}

@article{Wang2025,
  title={Stealthy false data injection attack scheduling design for multi-sensor systems with resource constraints},
  author={Wang, Ting-Ting and Yang, Guang-Hong and Dimirovski, Georgi Marko},
  journal={J. Franklin Inst.},
  volume={362},
  number={1},
  pages={107445},
  year={2025},
  publisher={Elsevier}
}

@article{Guo2023,
  title={Event-based optimal stealthy false data-injection attacks against remote state estimation systems},
  author={Guo, Haibin and Sun, Jian and Pang, Zhong-Hua and Liu, Guo-Ping},
  journal={IEEE Trans. Cybern.},
  volume={53},
  number={10},
  pages={6714--6724},
  year={2023},
  publisher={IEEE}
}

@inproceedings{biggs2024modeling,
  title={Modeling a heavy-duty vehicle data collection process},
  author={Biggs, Tyler and Lanigan, Trevor and Ruddell, David and Gallegos, Erika E and Daily, Jeremy},
  booktitle={Proc. 19th Annu. Syst. Syst. Eng. Conf. (SoSE)},
  pages={256--263},
  year={2024},
  organization={IEEE}
}

@article{zhao2021event,
  title={Event-triggered distributed fusion for multirate multisensor systems with heavy-tailed noises},
  author={Zhao, Ling and Cao, Xinyue and Li, Li and Yang, Hongjiu},
  journal={IEEE Trans. Syst., Man, Cybern., Syst.},
  volume={52},
  number={5},
  pages={3137--3150},
  year={2021},
  publisher={IEEE}
}
